\documentclass[prl,twocolumn,superscriptaddress]{revtex4-2}
\usepackage{bm}

\usepackage{qcircuit}
\usepackage{graphicx}
\usepackage{dcolumn}
\usepackage{braket}
\usepackage{amsmath}
\usepackage{color}
\usepackage{amssymb}
\usepackage{dsfont}
\usepackage{amsthm}
\usepackage{amsmath}
\usepackage{dsfont}
\usepackage{adjustbox}
\usepackage{tikz}
\usetikzlibrary{decorations.markings}
\usepackage{graphicx}
\usepackage{subcaption}
\usepackage{pgfplots}
\pgfplotsset{compat=1.3}
\usepgfplotslibrary{fillbetween}
\usetikzlibrary{patterns}
\usepackage{mathtools}
\usepackage{xcolor}
\usepackage{enumerate}
\usepackage{soul}
\usepackage[many]{tcolorbox} 		
\usepackage{physics}

\usepackage[labelfont=bf,
   justification=Justified,
   format=plain]{caption}
\usepackage[hidelinks]{hyperref}
\hypersetup{
    colorlinks,
    linkcolor={red!50!black},
    citecolor={blue!50!black},
    urlcolor={blue!80!black}
}

\usepackage{pgfplots}
\usepgfplotslibrary{groupplots}
\pgfplotsset{
  compat=1.18,
  clip mode=individual 
}

\definecolor{plasmaM1}{RGB}{13, 8, 135}
\definecolor{plasmaM2}{RGB}{83, 2, 163}
\definecolor{plasmaM3}{RGB}{139, 10, 165}
\definecolor{plasmaM4}{RGB}{186, 54, 135}
\definecolor{plasmaM5}{RGB}{219, 92, 104}
\definecolor{plasmaM6}{RGB}{244, 136, 73}
\definecolor{plasmaM7}{RGB}{253, 187, 45}
\definecolor{plasmaM8}{RGB}{240, 249, 33}

\definecolor{qdriftColor}{RGB}{13, 8, 135}
\definecolor{orderMColor}{RGB}{186, 54, 135}

\newcommand{\drawMregions}{%
  \fill[plasmaM1!20] (axis cs:0.0400, 1e-7) rectangle (axis cs:0.1857, 1.0);
  \fill[plasmaM2!20] (axis cs:0.1857, 1e-7) rectangle (axis cs:0.3369, 1.0);
  \fill[plasmaM3!20] (axis cs:0.3369, 1e-7) rectangle (axis cs:0.4877, 1.0);
  \fill[plasmaM4!20] (axis cs:0.4877, 1e-7) rectangle (axis cs:0.6370, 1.0);
  \fill[plasmaM5!20] (axis cs:0.6370, 1e-7) rectangle (axis cs:0.7850, 1.0);
  \fill[plasmaM6!20] (axis cs:0.7850, 1e-7) rectangle (axis cs:0.91, 1.0);
}

\newtheorem{lemma}{Lemma}
\newtheorem{theorem}{Theorem}
\newtheorem{proposition}{Proposition}
\newtheorem{corollary}{Corollary}[theorem]

\definecolor{color3}{RGB}{94, 201, 98}
\definecolor{color2}{RGB}{33, 145, 140}
\definecolor{color1}{RGB}{59, 82, 139}

\definecolor{plasma_color3}{RGB}{248, 149, 64}
\definecolor{plasma_color2}{RGB}{204, 71, 120}
\definecolor{plasma_color1}{RGB}{126, 3, 168}

\definecolor{colore}{rgb}{0.9, 0.9, 0.9}
\newtcolorbox{boxA}{
    colback = colore, 
    boxrule = 0pt  
}

\begin{document}

\title{Pathwise Random Hamiltonian Simulation}

\author{Davide Cugini}
\affiliation{Dipartimento di Fisica, Universit\`a di Pavia, via Bassi 6, 27100,  Pavia, Italy}

\begin{abstract}
Randomized product formulas such as qDrift offer a resource-efficient alternative to deterministic Trotter--Suzuki decompositions for Hamiltonian simulation, removing their polynomial dependence on the number of Hamiltonian terms. qDrift, however, is intrinsically limited to first order in the evolution time, so its query complexity remains linear in the inverse of the target accuracy, $1/\epsilon$. We introduce Pathwise Random Hamiltonian Simulation (PRHS), which extends qDrift to arbitrary order by subdividing each time step into $M$ correlated slices, each evolving under a term sampled from a quasi-probability distribution that we construct in closed form and prove unique, with a bias decaying factorially in $M$. Optimizing jointly over the number of slices $M$ and the number of independent blocks $N$ interpolates between the standard qDrift protocol at long times and a high-precision regime where the query cost grows slower than any power of $1/\epsilon$, without requiring ancillary qubits. Numerical simulations of five molecular Hamiltonians confirm this advantage, with PRHS achieving accuracies two to four orders of magnitude beyond qDrift at equal query cost.
\end{abstract}
\maketitle
\textit{Introduction}.---
Hamiltonian simulation is one of the central primitives of quantum computing. In its simplest form, it is the task of approximating, to within a target accuracy $\epsilon$, the time-evolution operator $U(t) = e^{-iHt}$ generated by a Hamiltonian $H$. Beyond its natural relevance for simulating quantum systems~\cite{piccinelli2025sqdrift}, Hamiltonian simulation constitutes the algorithmic backbone of numerous quantum algorithms, including phase estimation~\cite{kitaev1995quantum}, optimization~\cite{farhi2014quantum}, and linear-system solvers~\cite{harrow2009quantum}.

A widely used approach to this problem relies on \emph{operator-splitting} techniques, in which the evolution generated by a complicated Hamiltonian is approximated through compositions of evolutions generated by simpler Hamiltonian terms. More explicitly, given a Hamiltonian of the form
\begin{align}\label{eq:H_decomposition}
    H = \sum_{\gamma=1}^\Gamma \lambda_\gamma H_\gamma\,, \qquad \norm{H_\gamma} = 1\; \forall \gamma\,,
\end{align}
where $\Gamma$ is the number of elementary terms and $\Lambda := \sum_{\gamma} \abs{\lambda_\gamma}$ their total weight, these methods approximate $e^{-iHt}$ with products of the form $\prod_{l=1}^L e^{-it_l H_{\gamma_l}}$, for some sequence of times $t_l$ and Hamiltonian-term indices $\gamma_l$, a family that includes the well-known Trotter--Suzuki formulas~\cite{trotter1959product,suzuki1990fractal}. Such products reproduce $e^{-iHt}$ exactly only when the terms $H_\gamma$ mutually commute; in the generic case $[H_\gamma,H_{\gamma'}]\neq0$, they differ from the true evolution by an error controlled by nested commutators of the $H_\gamma$'s, and it is precisely this non-commutativity that operator-splitting methods must be designed to suppress. Such methods require no ancillary qubits or block encodings, and are conceptually simple and easy to implement.

Within this family, the $p$-th order Trotter--Suzuki formula achieves a query complexity scaling as $\mathcal{O}(5^{p}\Gamma^{2+1/p}t^{1+1/p}\epsilon^{-1/p}\max\abs{\lambda_\gamma}^{1+1/p})$~\cite{childs2021theory}, so that, for any fixed order $p$, the dependence on $1/\epsilon$ remains polynomial. Tuning $p$ as a function of both $\epsilon$ and $t$ allows product formulas to reach a quasi-optimal scaling in both quantities. In practice, however, this query complexity retains a polynomial dependence on $\Gamma$, a factor that becomes substantial for Hamiltonians with many non-commuting terms.

This observation motivated Campbell's proposal of the qDrift protocol~\cite{campbell2018random}: in this randomized version of product formulas, the evolution time is divided into $N$ equal time steps, but each step evolves under a single term $H_\gamma$, randomly sampled from a classical distribution $\mu_\gamma$. This simple procedure achieves the same asymptotic scaling in $t$ and $\epsilon$ as the first-order Trotter--Suzuki formula, while eliminating the explicit dependence on $\Gamma$.
qDrift, however, is intrinsically limited to first order in $t$, so its query complexity remains linear in $1/\epsilon$, just like fixed-order Trotter--Suzuki formulas. 
Several works have proposed extensions to higher order, each at the cost of one of qDrift's defining features: qSWIFT~\cite{nakaji2024qswift} reaches arbitrary order by importance-sampling a signed quasi-probability distribution over higher-order terms, but its estimation procedure requires an ancilla qubit, forfeiting qDrift's ancilla-free simplicity; composite qDrift-product formulas~\cite{pocrnic2024composite} remain ancilla-free by instead partitioning the Hamiltonian into a stochastically sampled part and a deterministically Trotterized part, but thereby reintroduce an explicit dependence on the number of terms $\Gamma$ through the Trotterized component; and, most recently, qSHIFT~\cite{lee2026qshift} adaptively updates the sampling distribution over blocks of $r$ terms to reach order $r$, at the cost of solving a system of $\Gamma^r$ classical equations at every round, an overhead that grows rapidly with the target order. An arbitrary-order generalization of qDrift that simultaneously preserves its ancilla-free simplicity, its independence from $\Gamma$, and an efficiently computable sampling rule has therefore remained an open problem. In this Letter we overcome this limitation.\\
\indent
We show that qDrift can be systematically extended to arbitrary order $M$ in $\Lambda t$ by replacing its independent, single-term sampling rule with a suitably correlated -- and, in general, signed -- quasi-probability distribution over sequences of $M$ Hamiltonian terms, a construction we call Pathwise Random Hamiltonian Simulation (PRHS). We prove that this distribution is unique, derive its closed-form expression, and establish a rigorous bias bound that decays factorially in $M$. We further show that, once the evolution time and target accuracy $\epsilon$ are fixed, the optimal choice of block size $M$ and number of repetitions $N$ interpolates between the standard qDrift protocol, recovered in the long-time limit, and a high-precision regime in which the query cost grows slower than any power of $1/\epsilon$, all while retaining the ancilla-free simplicity of randomized product formulas. We validate these predictions numerically on the dynamics of five molecular Hamiltonians, where PRHS achieves accuracies two to four orders of magnitude beyond qDrift at equal query cost per run.\\
\indent
We note, however, that our analysis is limited in scope, in that we only compare against other Trotter--Suzuki-type decompositions; approaches outside this family, most notably the quantum singular value transformation~\cite{gilyen2019qsvt,low2019qubitization}, achieve a better asymptotic scaling in both the evolution time and the target accuracy $\epsilon$, at the price of ancillary qubits and a more involved implementation.

\textit{Main results}.---
We work in the Liouvillian formulation, where
\begin{align}
    \mathcal{L}\rho := -i\left[H, \rho\right]\,,\qquad
     \mathcal{L}_\gamma \rho := -i\left[H_\gamma, \rho\right]
\end{align}
denote, respectively, the action of Liouvillian superoperators generated by $H$ and by each term $H_\gamma$ in the decomposition $H = \sum_{\gamma=1}^\Gamma \lambda_\gamma H_\gamma$ on a general density state $\rho$.
The qDrift protocol approximates the exact evolution $e^{t \mathcal{L}}$ to first order in $\Lambda t$ by sampling $\gamma$ from the discrete probability distribution
\begin{align}\label{eq:qdrift_distribution}
    \mu_\gamma = \frac{\abs{\lambda_\gamma}}{\Lambda}\,, \qquad \Lambda := \sum_{\gamma} \abs{\lambda_{\gamma}}\,,
\end{align}
and applying the corresponding evolution rule $e^{\Lambda t\,\mathcal{L}_\gamma}$.
The resulting quantum channel $\mathcal{C}_\mathrm{qDrift}(t) = \sum_\gamma \mu_\gamma\, e^{\Lambda t \mathcal{L}_\gamma}$ is a biased approximation of the true evolution, with a diamond-distance error bounded by $\delta \leq 2\Lambda^2t^2 e^{2\Lambda t} = \mathcal{O}((\Lambda t)^2)$ in the small $\Lambda t$ regime.

This work extends this construction to arbitrary order in $\Lambda t$. By subdividing each time step $t$ into $M$ \textit{time slices}, each associated with a randomly sampled evolution rule, and by carefully correlating the rules assigned to the $M$ slices within each step, we construct a channel
\begin{align}
    \mathcal{C}_M(t) := \sum_{\alpha \in \{1,...,\Gamma\}^M} q_M(\alpha) \,e^{\Lambda t \mathcal{L}_{\alpha_M}/M}\cdots e^{\Lambda t \mathcal{L}_{\alpha_1}/M}
\end{align}
that approximates $e^{t\mathcal{L}}$ to order $M$ in $\Lambda t$. We refer to this construction as pathwise random hamiltonian simulation.
Here, $\alpha=(\alpha_1,\ldots,\alpha_M)$ labels the trajectory of sampled indices, with $\alpha_m$ specifying the Hamiltonian term $H_{\alpha_m}$ applied in the $m$-th time slice.
In Appendix~A of the supplementary material~\cite{SM} we prove that the distribution $q_M(\alpha)$ achieving this scaling is unique, and derive its closed-form expression. 
In order to report such a result,
it is fundamental to observe that any index vector $\alpha=(\alpha_1,\ldots,\alpha_M)\in\{1,\ldots,\Gamma\}^M$ can be uniquely decomposed into $L\le M$ maximal \emph{runs} of consecutive equal entries, written $\alpha=(g_1^{c_1},\ldots,g_L^{c_L})$: here $g_l\in\{1,\ldots,\Gamma\}$ is the value taken by the $l$-th run, $c_l\ge1$ is the number of times $g_l$ occurs \emph{consecutively} (so $g_l^{c_l}$ denotes $c_l$ repetitions of $g_l$ in the string, not an algebraic power), consecutive runs have distinct values ($g_l\neq g_{l+1}$), and $\sum_{l=1}^L c_l=M$. Equality $L=M$ holds only when no two consecutive entries of $\alpha$ coincide; in general $L<M$ whenever some value repeats in a row. For instance, $\alpha=(1,1,2,3,3,3)$ has $M=6$ and decomposes into $L=3$ runs, $(g_1,c_1)=(1,2)$, $(g_2,c_2)=(2,1)$, $(g_3,c_3)=(3,3)$. Writing each index vector $\alpha\in\{1,\ldots,\Gamma\}^M$ in terms of the triple $(g,c,L)$ bijectively defined this way, the distribution reads
\begin{align}
    q_M(\alpha) \;=\;\sum_{b_1,\ldots,b_L\ge1}\frac{M^{\,n}}{n!}
    \prod_{l=1}^{L}\mu_{g_l}^{\,b_l}\,\frac{b_l!\;s(c_l,b_l)}{c_l!}\;,
    \label{eq:explicit}
\end{align}
where $n:=\sum_{l=1}^L b_l$ and $s(c,b)$ are the \emph{signed Stirling numbers} of the first kind. The distribution in Eq.~\eqref{eq:explicit} is the only quasi-probability distribution achieving a bias bounded by
\begin{align}
    \delta_M \leq \frac{(2\Lambda t)^{M+1}}{(M+1)!}\, e^{2\Lambda t}\,,
\end{align}
and it reduces to the usual qDrift sampling for $M=1$.
As in standard qDrift, an evolution over a total time $T=Nt$ is obtained by composing $N$ independent realizations of PRHS, so that the resulting channel is $\mathcal{C}_M(t)^N$, with total error bounded by
\begin{align}\label{eq:accuracy}
\delta(N,M) \leq N\frac{(2\Lambda T/N)^{M+1}}{(M+1)!}\,e^{2\Lambda T/N}\,.
\end{align}

At first sight, the increased precision afforded by PRHS comes at the cost of two main drawbacks relative to standard qDrift. The first is that each time step now requires $M$ queries to the elementary evolution channels, so the per-run query cost of $N$ independent applications of PRHS is $\mathrm{C}_\mathrm{run} := NM$.
The second is that, since $q_M$ is a quasi-probability distribution, i.e., $\sum_\alpha q_M(\alpha)=1$ while $q_M(\alpha)$ can take negative values, reproducing the target channel requires importance sampling: $\alpha$ is drawn from the normalized distribution $q_M(\alpha)/\norm{q_M}_1$, with $\norm{q_M}_1:=\sum_\alpha\abs{q_M(\alpha)}$, and the outcome of each of the $N$ independent applications is rescaled by $\norm{q_M}_1$. As a result, the variance of the estimator, and hence the number of shots needed to reach a target accuracy, is amplified by a factor $\norm{q_M}_1^{2N}$: the reduction in per-run cost is offset by an increase in the number of shots.
The magnitude of $\norm{q_M}_1$ depends strongly on the coefficients $\{\mu_\gamma\}_{\gamma=1}^\Gamma$.
In Appendix~B of~\cite{SM},
we prove the model-independent bound $\norm{q_M}_1\le\binom{2M-1}{M}$. The numerical results reported below further suggest that $\norm{q_M}_1$ grows exponentially with $M$, as $\norm{q_M}_1\sim\eta^{(M-1)/2}$ for a problem-specific rate $\eta$, consistent with the bound above, which restricts $\eta\in[1,16]$.
Remarkably, $\lVert q_M\rVert_1$ can be computed classically in advance,
before the quantum simulation is run.

However, a deeper investigation reveals that these shortcomings are less severe than expected.
Once the computational problem is defined, $\Lambda$, $T$, and the target accuracy $\epsilon$ are fixed.
At this stage, an optimal choice of the integer pair $(N,M)$ is determined by two conditions.
First, the accuracy $\delta(N,M)$ in Eq.~\eqref{eq:accuracy} must satisfy $\delta \leq \epsilon$.
Second, the remaining flexibility can be used to minimize the total computational cost
\begin{align}\label{eq:cost}
    \mathrm{C}(N,M) := (NM) \times \eta^{N(M-1)},
\end{align}
where the first factor on the r.h.s. is the query cost per run $\mathrm{C}_\mathrm{run} = NM$, while the second accounts for the amplification in the number of shots.
We provide a detailed description of how such an optimization can be efficiently performed numerically in Appendix~C of~\cite{SM}.
Here, we only report the optimal choices in two specific limits.
The first is the long-time limit, $\Lambda T \to \infty$ at fixed $\epsilon$.
In this case, the optimal choice reduces to $M^* = 1$ and $N^* \sim 2\Lambda^2T^2/\epsilon$, which coincides with the standard qDrift protocol.
The second is the high-accuracy limit, $\epsilon \to 0$ at fixed $\Lambda T$.
In this case, the optimal choice goes in exactly the opposite direction, i.e., $N^* = 1$ and $M^* \sim \ln(1/\epsilon)/\ln(\ln(1/\epsilon))$, meaning that all the time slices of the computation are mutually correlated.
The optimal cost achieved in these two limits can be directly computed from Eq.~\eqref{eq:cost}.
In particular, in the long-time limit there is no need for shot-number amplification, and the cost is $\mathrm{C}(N^*,M^*) = N^* = \mathcal{O}(\Lambda^2T^2/\epsilon)$, which is quadratically suboptimal, though still polynomial.
In the high-precision setting, instead, the optimal cost per run is
\begin{align}
    \mathrm{C}_\mathrm{run} = M^* \sim \frac{\ln(1/\epsilon)}{\ln\ln(1/\epsilon)}\,,
\end{align}
while the shots amplification is $\eta^{M^*-1}$, resulting in a total cost scaling as
\begin{align}
    \mathrm{C}(N^*,M^*) \sim \left(\frac{1}{\epsilon}\right)^{\frac{\ln\eta}{\ln\ln(1/\epsilon)}}\,,
\end{align}
which is slower than any power of $1/\epsilon$, and hence PRHS outperforms qDrift.
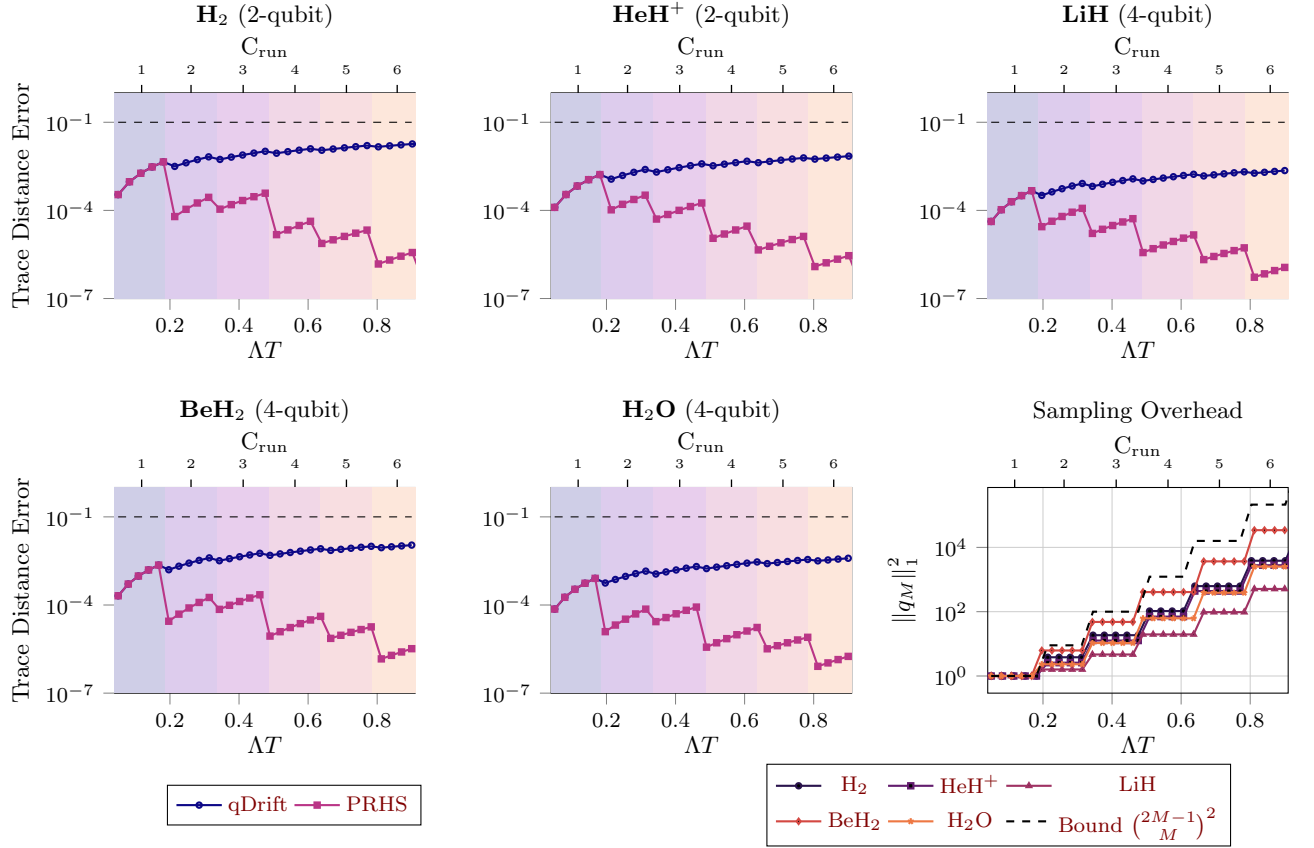
\begin{figure*}[t!]
  \centering
  \definecolor{molH2}{RGB}{31,119,180}
  \definecolor{molHeH}{RGB}{255,127,14}
  \definecolor{molLiH}{RGB}{44,160,44}
  \definecolor{molBeH2}{RGB}{214,39,40}
  \definecolor{molH2O}{RGB}{148,103,189}
  \begin{tikzpicture}
    \begin{groupplot}[
      group style={
        group size=3 by 2,
        horizontal sep=1.8cm,
        vertical sep=2.5cm,
      },
      width=0.31\textwidth,
      height=0.24\textwidth,
      xmin=0.04,
      xmax=0.91,
      xlabel={$\Lambda T$},
      xlabel style={yshift=10pt},
      xtick pos=bottom,
      ytick pos=left,
      tick align=outside,
      title style={yshift=15pt},
      grid=both,
      grid style={line width=.1pt, draw=gray!25},
      major grid style={line width=.2pt, draw=gray!40},
      tick label style={font=\footnotesize},
      label style={font=\small},
      legend style={font=\footnotesize},
      extra x ticks={0.118, 0.261, 0.412, 0.562, 0.711, 0.858},
      extra x tick labels={$1$, $2$, $3$, $4$, $5$, $6$},
      extra x tick style={
        tick pos=top,
        xticklabel pos=top,
        tick align=outside,
        major tick length=3pt,
        tick style={draw=black, line width=0.5pt},
        grid=none,
        tick label style={font=\tiny, yshift=1pt}
      }
    ]
    \nextgroupplot[title={\small \textbf{H}$_2$ (2-qubit)}, ymode=log, ymin=1e-7, ymax=1.0, ylabel={Trace Distance Error}]
      \drawMregions
      \node[font=\small] at (axis description cs:0.5, 1.22) {$\mathrm{C}_\mathrm{run}$};
      \draw[black, dashed] (axis cs:0.05, 1e-1) -- (axis cs:0.9, 1e-1);
      \addplot[color=qdriftColor, mark=o, mark size=1.0pt, thick] table [x=s, y=qdrift_error, col sep=comma] {data_H2.csv};
      \addplot[color=orderMColor, mark=square*, mark size=1.0pt, thick] table [x=s, y=order_m_error, col sep=comma] {data_H2.csv};
    \nextgroupplot[title={\small \textbf{HeH}$^+$ (2-qubit)}, ymode=log, ymin=1e-7, ymax=1.0]
      \drawMregions
      \node[font=\small] at (axis description cs:0.5, 1.22) {$\mathrm{C}_\mathrm{run}$};
      \draw[black, dashed] (axis cs:0.05, 1e-1) -- (axis cs:0.9, 1e-1);
      \addplot[color=qdriftColor, mark=o, mark size=1.0pt, thick] table [x=s, y=qdrift_error, col sep=comma] {data_HeH+.csv};
      \addplot[color=orderMColor, mark=square*, mark size=1.0pt, thick] table [x=s, y=order_m_error, col sep=comma] {data_HeH+.csv};
    \nextgroupplot[title={\small \textbf{LiH} (4-qubit)}, ymode=log, ymin=1e-7, ymax=1.0]
      \drawMregions
      \node[font=\small] at (axis description cs:0.5, 1.22) {$\mathrm{C}_\mathrm{run}$};
      \draw[black, dashed] (axis cs:0.05, 1e-1) -- (axis cs:0.9, 1e-1);
      \addplot[color=qdriftColor, mark=o, mark size=1.0pt, thick] table [x=s, y=qdrift_error, col sep=comma] {data_LiH.csv};
      \addplot[color=orderMColor, mark=square*, mark size=1.0pt, thick] table [x=s, y=order_m_error, col sep=comma] {data_LiH.csv};
    \nextgroupplot[title={\small \textbf{BeH}$_2$ (4-qubit)}, ymode=log, ymin=1e-7, ymax=1.0, ylabel={Trace Distance Error}]
      \drawMregions
      \node[font=\small] at (axis description cs:0.5, 1.22) {$\mathrm{C}_\mathrm{run}$};
      \draw[black, dashed] (axis cs:0.05, 1e-1) -- (axis cs:0.9, 1e-1);
      \addplot[color=qdriftColor, mark=o, mark size=1.0pt, thick] table [x=s, y=qdrift_error, col sep=comma] {data_BeH2.csv};
      \addplot[color=orderMColor, mark=square*, mark size=1.0pt, thick] table [x=s, y=order_m_error, col sep=comma] {data_BeH2.csv};
    \nextgroupplot[
      title={\small \textbf{H}$_2$\textbf{O} (4-qubit)},
      ymode=log,
      ymin=1e-7,
      ymax=1.0,
      legend to name=sharedlegend_M,
      legend columns=2
    ]
      \drawMregions
      \node[font=\small] at (axis description cs:0.5, 1.22) {$\mathrm{C}_\mathrm{run}$};
      \draw[black, dashed] (axis cs:0.05, 1e-1) -- (axis cs:0.9, 1e-1);
      \addplot[color=qdriftColor, mark=o, mark size=1.0pt, thick] table [x=s, y=qdrift_error, col sep=comma] {data_H2O.csv};
      \addlegendentry{qDrift}
      \addplot[color=orderMColor, mark=square*, mark size=1.0pt, thick] table [x=s, y=order_m_error, col sep=comma] {data_H2O.csv};
      \addlegendentry{PRHS}
    \nextgroupplot[
      title={\small Sampling Overhead},
      ymode=log,
      ylabel={$\norm{q_M}_1^2$},
      legend to name=sharedlegend_norm,
      legend columns=3
    ]
      \node[font=\small] at (axis description cs:0.5, 1.2) {$\mathrm{C}_\mathrm{run}$};
    \addplot[draw={rgb,255:red,42;green,20;blue,77}, mark=*, mark size=1.0pt, thick]
        table [x=s, y expr={(\thisrow{norm_q1})^2}, col sep=comma] {data_H2.csv};
    \addlegendentry{$\text{H}_2$}
    
    \addplot[draw={rgb,255:red,99;green,29;blue,110}, mark=square*, mark size=1.0pt, thick]
        table [x=s, y expr={(\thisrow{norm_q1})^2}, col sep=comma] {data_HeH+.csv};
    \addlegendentry{$\text{HeH}^+$}
    
    \addplot[draw={rgb,255:red,159;green,50;blue,99}, mark=triangle*, mark size=1.0pt, thick]
        table [x=s, y expr={(\thisrow{norm_q1})^2}, col sep=comma] {data_LiH.csv};
    \addlegendentry{$\text{LiH}$}
    
    \addplot[draw={rgb,255:red,212;green,72;blue,66}, mark=diamond*, mark size=1.0pt, thick]
        table [x=s, y expr={(\thisrow{norm_q1})^2}, col sep=comma] {data_BeH2.csv};
    \addlegendentry{$\text{BeH}_2$}
    
    \addplot[draw={rgb,255:red,238;green,120;blue,67}, mark=star, mark size=1.2pt, thick]
        table [x=s, y expr={(\thisrow{norm_q1})^2}, col sep=comma] {data_H2O.csv};
    \addlegendentry{$\text{H}_2\text{O}$}
    
    \addplot[black, dashed, thick, no markers]
        table [x=s, y expr={(\thisrow{bound_q1})^2}, col sep=comma] {data_bound.csv};
    \addlegendentry{Bound $\binom{2M-1}{M}^2$}
    \end{groupplot}
  \end{tikzpicture}
  \vspace{3mm}
  \noindent
  \begin{minipage}[c]{0.48\textwidth}
    \centering
    \ref{sharedlegend_M}
  \end{minipage}%
  \hfill
  \begin{minipage}[c]{0.48\textwidth}
    \centering
    \ref{sharedlegend_norm}
  \end{minipage}
  \caption{Comparison between qDrift and PRHS ($N=1$) for the five
molecular Hamiltonians introduced in the text. The first 5 panels show the trace-distance error
achieved by qDrift (blue) and by PRHS (purple) as a function of
$\Lambda T$, at matched query cost $\mathrm{C}_\mathrm{run}$; the black dashed horizontal
line marks the target accuracy $\epsilon=10^{-1}$ used to select, for each $\Lambda T$, the smallest
$M$ meeting this target. Shaded background regions, and the labels along the upper
horizontal axis, indicate the resulting order $M$, which, following from the
bias bound of Eq.~\eqref{eq:accuracy} evaluated at $N=1$, depends on $\Lambda T$ alone
and is therefore common to all five molecules. 
The last panel reports the corresponding sampling
overhead $\lVert q_M \rVert_1^2$, computed exactly from Eq.~\eqref{eq:explicit} for each
molecule (colored curves, molecule-specific through $\mu_\gamma=\abs{\lambda_\gamma}/\Lambda$),
together with the universal, molecule-independent bound $\binom{2M-1}{M}^2$ of Appendix~B~\cite{SM} (black dashed) that every curve must lie below; the same
top-axis $M$ labeling applies.
Although $\norm{q_M}_1^2$ grows exponentially with
$M$, the resulting increase in sampling cost remains modest compared to the magnitude gain in accuracy.}
\label{fig:molecular_benchmark_M_regions}
\end{figure*}

\textit{Numerical results}.---
Since PRHS and qDrift coincide exactly in the long-time limit $\Lambda T \to \infty$ (where $M^*=1$), a meaningful numerical comparison must instead target the regime in which correlating multiple slices, $M>1$, becomes necessary to reach a fixed accuracy -- the same mechanism underlying the high-accuracy limit $\epsilon\to0$ derived above. We therefore fix $\epsilon = 10^{-1}$ and scan $\Lambda T \in [0,0.9]$, comparing the $N=1$ instance of PRHS against qDrift, allowing $M$ to grow with $\Lambda T$ as needed to reach the target accuracy. For each value of $\Lambda T$, we take the smallest $M$ (and hence $\mathrm{C}_\mathrm{run}=M$) meeting the target accuracy, and run qDrift with a matching number of steps equal to this same $\mathrm{C}_\mathrm{run}$, so that PRHS and qDrift are compared at equal query cost.
We evaluate this benchmark on five molecular Hamiltonians built in the STO-3G basis~\cite{hehre1969selfconsistent}: $\mathrm{H_2}$ and $\mathrm{HeH^+}$ are retained at the full STO-3G level, while $\mathrm{LiH}$, $\mathrm{BeH_2}$, and $\mathrm{H_2O}$ are reduced to a complete active space~\cite{roos1980complete} of two electrons in two spatial orbitals [$\mathrm{CAS}(2,2)$] under the frozen-core approximation. Under the Jordan--Wigner transformation~\cite{jordan1928paulische}, these representations map onto $2$, $2$, $4$, $4$, and $4$ qubits, respectively. For each system, time evolution starts from the single-determinant Hartree-Fock (HF) reference state, the standard mean-field baseline from which real-time electron-correlation dynamics develop, corresponding to the computational-basis states $\ket{00}$, $\ket{01}$, $\ket{0011}$, $\ket{1100}$, and $\ket{1100}$, respectively.
Panels 1--5 of Fig.~\ref{fig:molecular_benchmark_M_regions} show that, in all five cases, both PRHS and qDrift remain below the target accuracy $\epsilon$ over the whole range. For small $\Lambda T$, the minimal query cost already saturates at $M=1$, where PRHS coincides exactly with qDrift; as $\Lambda T$ grows and larger $M$ becomes necessary, the error achieved by PRHS falls two to four orders of magnitude below that of qDrift at matched query cost.
The last panel of Fig.~\ref{fig:molecular_benchmark_M_regions} complements this comparison by reporting $\norm{q_M}_1$ over the same range of $\Lambda T$ for all five systems: although $\norm{q_M}_1$ grows exponentially with $M$, as established above, this growth remains small compared to the corresponding gain in precision, and lies safely below the universal bound $\binom{2M-1}{M}$.

\textit{Conclusion}.---
We have introduced pathwise random hamiltonian simulation, a systematic extension of the qDrift protocol to arbitrary order $M$ in the evolution time, obtained by replacing qDrift's independent, single-term sampling rule with a correlated quasi-probability distribution over sequences of $M$ Hamiltonian terms. We constructed this distribution in closed form, proved it unique, and established a bias bound that decays factorially in $M$.

Being based on a signed quasi-probability distribution, however, the method carries an intrinsic drawback shared by all such randomized schemes: reconstructing the target channel requires importance sampling, and rescaling measurement outcomes amplifies the number of shots needed to reach a target accuracy. We showed, nonetheless, that this cost can be kept under control: jointly optimizing over the block size $M$ and the number of independent repetitions $N$ allows PRHS to interpolate between the standard qDrift protocol in the long-time regime and a high-precision regime in which the total query cost grows slower than any power of $1/\epsilon$, despite the exponential growth of the shots-number amplification with $M$. Numerical simulations on the dynamics of five molecular Hamiltonians confirm the practical value of this high-precision regime: at equal query cost, PRHS reaches accuracies two to four orders of magnitude beyond qDrift.

As a final remark,
the $(N,M)$ optimization presented here targets the abstract query cost $\mathrm{C}_\mathrm{run}=NM$, which counts only the number of elementary evolution channels $e^{\Lambda t \mathcal{L}_\gamma/M}$ applied, treating every $\gamma$ as equally expensive. On real hardware, however, the implementation cost of $e^{\Lambda t \mathcal{L}_\gamma/M}$ generically varies across terms and depends on the target architecture. The resource-optimal importance-sampling framework of Ref.~\cite{cugini2026resource} is built precisely for this setting: given the distribution of implementation costs of $e^{\Lambda t \mathcal{L}_\gamma}$ across $\gamma$ on a specific architecture, it jointly minimizes a net-cost figure of merit combining these architecture-dependent costs with the resulting estimator variance, reweighting the sampling distribution over $\gamma$ away from the naive choice $\propto\abs{q_M(\alpha)}$ used here, while provably preserving the estimator's bias. Combining this architecture-aware reweighting with the $(N,M)$ optimization of Eq.~\eqref{eq:cost} is a natural next step towards minimizing the true, hardware-level computational cost of PRHS, and, more broadly, of randomized Hamiltonian-simulation protocols beyond qDrift.

\textit{Acknowledgments}.---
The author thanks Luca Spagnoli for stimulating discussions and helpful feedback and Sara Molteni for her constant support throughout this work.

\bibliography{main.bib}

\widetext

\appendix

\clearpage
\begin{center}
\textbf{\large Supplementary Material}\\
\end{center}
\setcounter{equation}{0}
\setcounter{figure}{0}
\setcounter{table}{0}
\setcounter{page}{1}
\makeatletter
\renewcommand{\theequation}{S\arabic{equation}}
\renewcommand{\thefigure}{S\arabic{figure}}

In this Supplementary Material we provide all the detailed proofs of the statements of the main text.

\section{A --- Existence and uniqueness of the quasi-probability distribution}\label{app:existence_uniqueness}
Throughout this appendix, the elementary time step $\Lambda t/M$ has been absorbed into the definition of $\mathcal{L}_1,\ldots,\mathcal{L}_\Gamma$, i.e.\ these are understood to be expressed in units of $(\Lambda t/M)^{-1}$ -- equivalently, $\mathcal{L}_\gamma\equiv(\Lambda t/M)\,\mathcal{L}_\gamma^{\rm phys}$ for the physical (dimensionful) generators, so that $e^{M\mathcal{L}}$ below reproduces the physical evolution over the full time $t$.

Let $\Omega:=\{1,\ldots,\Gamma\}$ index a set of Liouvillians $\mathcal{L}_1,\ldots,\mathcal{L}_\Gamma$, weighted by $\mu_\gamma:= \abs{\lambda_\gamma}/\Lambda>0$ so that $\sum_{\gamma\in\Omega}\mu_\gamma=1$, and $\mathcal{L}:=\sum_{\gamma\in\Omega}\mu_\gamma\mathcal{L}_\gamma$. A string $\alpha=(\alpha_1,\ldots,\alpha_M)\in\Omega^M$ realizes the operator
\begin{align}
    E_\alpha:=e^{\mathcal{L}_{\alpha_M}}\cdots e^{\mathcal{L}_{\alpha_1}}
    \label{eq:E_alpha}
\end{align}
(the channel labeled $\alpha_1$ is applied first, hence sits rightmost).

Let $\mathds{A}$ be the associative, in general noncommutative, algebra generated by $\mathcal{L}_1,\ldots,\mathcal{L}_\Gamma$, completed so as to contain formal series such as $e^{\mathcal{L}_\gamma}$ and let $\mathds{A}_{\ge n}\subset\mathds{A}$ be the subspace spanned by monomials $\mathcal{L}_{\delta_1}\cdots\mathcal{L}_{\delta_k}$ of degree $k\ge n$, so that $\mathds{A}_{\ge a}\cdot\mathds{A}_{\ge b}\subseteq\mathds{A}_{\ge a+b}$.

A quasi-probability $q:\Omega^M\to\mathbb{R}$ realizes an \emph{order-$M$} sampling scheme if
\begin{align}
    \sum_{\alpha\in\Omega^M} q(\alpha)\,E_\alpha\ \equiv\ e^{M\mathcal{L}}
    \pmod{\mathds{A}_{\ge M+1}}\,,
    \label{eq:kth_order_constraint}
\end{align}
i.e.\ the sampled channel reproduces $e^{M\mathcal{L}}$ to every order $\le M$ in the generators. Below we establish both halves of this statement: that a solution $q$ exists for every $M$ and $\Gamma$, given in closed form by Theorem~\ref{thm:closed_form}, and that this solution is the \emph{only} one, proved separately in Theorem~\ref{prop:uniqueness}.

For a string $\alpha=(\alpha_1,\ldots,\alpha_M)\in\Omega^M$ we write its \emph{run decomposition} as $(g_1,c_1),\ldots,(g_L,c_L)$, where $g_l\ne g_{l+1}$ are the run values and $c_l\ge1$ the run lengths, $\sum_l c_l=M$.

\begin{theorem}[Closed form for $q$]
\label{thm:closed_form}
The quasi-probability $q$ solving the order-$M$ condition \eqref{eq:kth_order_constraint} is
\begin{align}
    \boxed{\;q(\alpha)\;=\;\sum_{b_1,\ldots,b_L\ge1}\frac{M^{\,n}}{n!}
    \prod_{l=1}^{L}\mu_{g_l}^{\,b_l}\,\frac{b_l!\;s(c_l,b_l)}{c_l!}\;,
    \qquad n:=\sum_{l=1}^L b_l\;,}
\end{align}
where $s(c,b)$ are the \emph{signed} Stirling numbers of the first kind, defined by $\sum_b s(c,b)x^b=x(x-1)\cdots(x-c+1)$. Equivalently, with
\begin{align}
    a_c(\mu,z)\;:=\;\frac{1}{c!}\sum_{b=1}^{c}s(c,b)\,b!\,\mu^{b}\,z^{\,c-b}
    \qquad\text{and}\qquad
    \langle z^{D}\rangle_M:=\frac{M^{\underline{D}}}{M^{D}}\,,\ \ M^{\underline{D}}:=M(M-1)\cdots(M-D+1)\,,
    \label{eq:a_poly}
\end{align}
\begin{align}
    q(\alpha)\;=\;\frac{M^M}{M!}\,\Big\langle\ \prod_{l=1}^{L}a_{c_l}\!\left(\mu_{g_l},z\right)\Big\rangle_M \,.
    \label{eq:pairing}
\end{align}
\end{theorem}

The first few run polynomials are
\begin{align}
    a_1=\mu\,,\quad
    a_2=\mu\!\left(\mu-\tfrac{z}{2}\right)\,,\quad
    a_3=\mu\!\left(\mu^2-\mu z+\tfrac{z^2}{3}\right)\,,\quad
    a_4=\mu\!\left(\mu^3-\tfrac32\mu^2z+\tfrac{11}{12}\mu z^2-\tfrac14 z^3\right)\,.
    \label{eq:a_list}
\end{align}

Equation~\eqref{eq:pairing} depends on $\alpha$ only through its run data $\{(g_l,c_l)\}$; setting $\langle 1\rangle_M=1$, $\langle z\rangle_M=1$, $\langle z^2\rangle_M=\frac{M-1}{M}$ recovers, as one-line special cases, the no-repeat word $q=\frac{M^M}{M!}\prod_j\mu_{\gamma_j}$, an isolated pair $q=\frac{M^M}{M!}\big(\prod_l\mu_{g_l}\big)(\mu_{g_i}-\tfrac12)$, an isolated triple $q=\frac{M^M}{M!}\big(\prod_l\mu_{g_l}\big)(\mu_{g_i}^2-\mu_{g_i}+\tfrac{M-1}{3M})$, two isolated pairs $q=\frac{M^M}{M!}\big(\prod_l\mu_{g_l}\big)\big[(\mu_{g_i}-\tfrac12)(\mu_{g_j}-\tfrac12)-\tfrac{1}{4M}\big]$, and the all-equal word $q=\binom{M\mu}{M}$.

\subsubsection{Proof of existence}
We show that \eqref{eq:explicit} satisfies \eqref{eq:kth_order_constraint}. Throughout, $\alpha\leftrightarrow(g_l,c_l)_{l=1}^L$ denotes the run decomposition introduced above, and we freely rename the run values $g_l$ as $\gamma_l$.

We start from $e^{M\mathcal{L}}=e^{M\sum_{\gamma\in\Omega}\mu_\gamma\mathcal{L}_\gamma}$, and, for a central (scalar) formal parameter $\rho$, define
\begin{align}
    W(\rho):=\sum_{\gamma\in\Omega}\mu_\gamma\ln\!\big(1+\rho[e^{\mathcal{L}_\gamma}-1]\big)
    \ =:\ \sum_{n\ge1}\rho^n W_n\ \in\ \mathds{A}[[\rho]]\,.
    \label{eq:Wrho}
\end{align}
Since $1+1\cdot[e^{\mathcal{L}_\gamma}-1]=e^{\mathcal{L}_\gamma}$, we have $W(1)=\sum_\gamma\mu_\gamma\mathcal{L}_\gamma=\mathcal{L}$, hence $e^{M\mathcal{L}}=e^{MW(1)}$. Because $e^{\mathcal{L}_\gamma}-1\in\mathds{A}_{\ge1}$ and $\ln(1+\rho X)=\sum_{k\ge1}\frac{(-1)^{k-1}}{k}\rho^kX^k$, each coefficient $W_n=\sum_\gamma\mu_\gamma\frac{(-1)^{n-1}}{n}[e^{\mathcal{L}_\gamma}-1]^n$ of \eqref{eq:Wrho} lies in $\mathds{A}_{\ge n}$. Writing $e^{MW(\rho)}=:\sum_{n\ge0}\rho^n A_n$ and expanding the exponential,
\begin{align}
    A_n=\sum_{j\ge0}\frac1{j!}\sum_{\substack{k_1,\ldots,k_j\ge1\\k_1+\cdots+k_j=n}}
    (MW_{k_1})\cdots(MW_{k_j})
\end{align}
is a finite sum of products of elements of $\mathds{A}_{\ge k_1},\ldots,\mathds{A}_{\ge k_j}$, hence lies in $\mathds{A}_{\ge k_1+\cdots+k_j}=\mathds{A}_{\ge n}$: $A_n\in\mathds{A}_{\ge n}$ for every $n\ge0$, and $e^{M\mathcal{L}}=\sum_{n\ge0}A_n$.

To isolate the terms up to order $M$, change variables to $\rho:=y(1+y)^{-1}$, so that $\rho^n=y^n(1+y)^{-n}$, and define
\begin{align}
    \mathcal{G}(y):=(1+y)^M\,e^{MW(\rho)}=\sum_{n\ge0} y^n(1+y)^{M-n}\,A_n\,.
    \label{eq:Gdef}
\end{align}
For $0\le n\le M$, $(1+y)^{M-n}$ is a polynomial of degree $M-n$, so $y^n(1+y)^{M-n}$ contributes only to powers of $y$ up to $y^M$; for $n>M$, $(1+y)^{M-n}=\sum_{j\ge0}\binom{M-n}{j}y^j$ is a power series starting at $y^0$, so $y^n(1+y)^{M-n}$ contributes nothing below $y^n>y^M$. Hence
\begin{align}
    [y^M]\,\mathcal{G}(y)=\sum_{n=0}^{M}A_n\,,
    \label{eq:yMextract}
\end{align}
and since $A_n\in\mathds{A}_{\ge n}$, we have $e^{M\mathcal{L}}-[y^M]\mathcal{G}(y)=\sum_{n>M}A_n\in\mathds{A}_{\ge M+1}$, i.e.
\begin{align}
    [y^M]\,\mathcal{G}(y)\ \equiv\ e^{M\mathcal{L}}\pmod{\mathds{A}_{\ge M+1}}\,.
    \label{eq:yMcongruence}
\end{align}

It remains to evaluate $[y^M]\mathcal{G}(y)$ explicitly. For every $\gamma$, under $\rho=y(1+y)^{-1}$,
\begin{align}
    1+\rho[e^{\mathcal{L}_\gamma}-1]
    =\frac{(1+y)+y[e^{\mathcal{L}_\gamma}-1]}{1+y}
    =\frac{1+ye^{\mathcal{L}_\gamma}}{1+y}\,,
\end{align}
so $\ln(1+\rho[e^{\mathcal{L}_\gamma}-1])=\ln(1+ye^{\mathcal{L}_\gamma})-\ln(1+y)$. Summing against $\mu_\gamma$ and using $\sum_\gamma\mu_\gamma=1$ — this is exactly where the normalization of $\mu$ enters —
\begin{align}
    MW(\rho)=M\sum_{\gamma\in\Omega}\mu_\gamma\ln(1+ye^{\mathcal{L}_\gamma})-M\ln(1+y)\,,
\end{align}
so the factor $(1+y)^{-M}=e^{-M\ln(1+y)}$ cancels the prefactor $(1+y)^M$ in \eqref{eq:Gdef} exactly, leaving
\begin{align}
    \mathcal{G}(y)=\exp\!\Big(M\sum_{\gamma\in\Omega}\mu_\gamma\ln(1+ye^{\mathcal{L}_\gamma})\Big)\,.
    \label{eq:Gclean}
\end{align}

We now expand $\mathcal{G}(y)$ into words. Writing $X_\gamma:=\mu_\gamma\ln(1+ye^{\mathcal{L}_\gamma})$, we have $\mathcal{G}(y)=\sum_{n\ge0}\frac{M^n}{n!}\big(\sum_{\gamma\in\Omega}X_\gamma\big)^n$. The sum $\sum_{\delta\in\Omega^n}X_{\delta_1}\cdots X_{\delta_n}$ ranges over all of $\Omega^n$, so relabelling $\delta\mapsto(\delta_n,\ldots,\delta_1)$ — a bijection of $\Omega^n$ with itself — shows
\begin{align}
    \Big(\sum_{\gamma\in\Omega}X_\gamma\Big)^n
    =\sum_{\delta\in\Omega^n}X_{\delta_n}X_{\delta_{n-1}}\cdots X_{\delta_1}\,.
\end{align}
Grouping each maximal run of equal consecutive entries of $\delta=(\delta_1,\ldots,\delta_n)$ into a block identifies $\Omega^n$ with the disjoint union, over $L\ge0$ and distinct-adjacent $\gamma_1,\ldots,\gamma_L\in\Omega$ ($\gamma_l\ne\gamma_{l+1}$), of compositions $b\in\mathbb{Z}_{\ge1}^L$ with $|b|:=\sum_l b_l=n$, via $\delta=(\gamma_1^{b_1},\ldots,\gamma_L^{b_L})$; under this identification $X_{\delta_n}\cdots X_{\delta_1}=X_{\gamma_L}^{b_L}\cdots X_{\gamma_1}^{b_1}=:\overset{\leftarrow}{\prod_{l=1}^L}X_{\gamma_l}^{b_l}$, where $\overset{\leftarrow}{\prod}_{l=1}^L Y_l:=Y_L Y_{L-1}\cdots Y_1$ denotes the product taken in \emph{decreasing} order of $l$. Hence
\begin{align}
    \mathcal{G}(y)=\sum_{L\ge0}\ \sum_{\substack{\gamma_1,\ldots,\gamma_L\in\Omega\\\gamma_l\ne\gamma_{l+1}}}
    \ \sum_{n\ge0}\frac{M^n}{n!}\sum_{\substack{b\in\mathbb{Z}_{\ge1}^L\\|b|=n}}
    \ \overset{\leftarrow}{\prod_{l=1}^L}\big(\mu_{\gamma_l}\ln(1+ye^{\mathcal{L}_{\gamma_l}})\big)^{b_l}\,.
    \label{eq:Gexpand1}
\end{align}
By the classical exponential generating function of the signed Stirling numbers of the first kind, $\ln^b(1+x)=b!\sum_{c\ge b}s(c,b)\,x^c/c!$, applied with $x=ye^{\mathcal{L}_{\gamma_l}}$ (an ordinary power series, as it involves only the single generator $\mathcal{L}_{\gamma_l}$),
\begin{align}
    \ln^{b_l}(1+ye^{\mathcal{L}_{\gamma_l}})
    =b_l!\sum_{c_l\ge b_l}\frac{s(c_l,b_l)}{c_l!}\,y^{c_l}e^{c_l\mathcal{L}_{\gamma_l}}\,,
\end{align}
so, carrying out the ordered product,
\begin{align}
    \overset{\leftarrow}{\prod_{l=1}^L}\big(\mu_{\gamma_l}\ln(1+ye^{\mathcal{L}_{\gamma_l}})\big)^{b_l}
    =\sum_{c_1,\ldots,c_L\ge1}\prod_{l=1}^L\Big(\mu_{\gamma_l}^{b_l}\,\frac{b_l!\,s(c_l,b_l)}{c_l!}\Big)
    \,y^{c_l}\cdot e^{c_L\mathcal{L}_{\gamma_L}}\cdots e^{c_1\mathcal{L}_{\gamma_1}}\,.
\end{align}
Identifying each term with the word $\alpha:=(\gamma_1^{c_1},\ldots,\gamma_L^{c_L})\in\Omega^N$, $N:=\sum_l c_l$, whose run decomposition is exactly $(\gamma_l,c_l)_{l=1}^L$, we have $\prod_l y^{c_l}=y^N$ and $e^{c_L\mathcal{L}_{\gamma_L}}\cdots e^{c_1\mathcal{L}_{\gamma_1}}=E_\alpha$. Substituting into \eqref{eq:Gexpand1} and collecting, for each fixed $\alpha$, all terms over $b_1,\ldots,b_L\ge1$ (the constraint $b_l\le c_l$ is automatic since $s(c_l,b_l)=0$ for $b_l>c_l$),
\begin{align}
    \mathcal{G}(y)=\sum_{N\ge0} y^N\sum_{\alpha\in\Omega^N} q(\alpha)\,E_\alpha\,,
    \qquad q(\alpha):=\sum_{b_1,\ldots,b_L\ge1}\frac{M^{|b|}}{|b|!}
    \prod_{l=1}^L\mu_{\gamma_l}^{b_l}\,\frac{b_l!\,s(c_l,b_l)}{c_l!}\,,
    \label{eq:Gwords}
\end{align}
which is exactly \eqref{eq:explicit}. Reading off the coefficient of $y^M$ in \eqref{eq:Gwords},
\begin{align}
    [y^M]\,\mathcal{G}(y)=\sum_{\alpha\in\Omega^M} q(\alpha)\,E_\alpha\,,
\end{align}
which, combined with \eqref{eq:yMcongruence}, gives exactly \eqref{eq:kth_order_constraint} for $q$ as in \eqref{eq:explicit}. This proves that \eqref{eq:explicit} solves the order-$M$ condition, for every $M$ and $\Gamma$.

It remains to check that \eqref{eq:explicit} coincides with the equivalent form \eqref{eq:pairing}. Fix $\alpha$ with run decomposition $(g_l,c_l)_{l=1}^L$, $\sum_l c_l=M$. Expanding the definition of $a_c(\mu,z)$ and collecting powers of $z$,
\begin{align}
    \prod_{l=1}^L a_{c_l}(\mu_{g_l},z)
    =\prod_{l=1}^L\left[\frac{1}{c_l!}\sum_{b_l=1}^{c_l}s(c_l,b_l)\,b_l!\,\mu_{g_l}^{b_l}\,z^{c_l-b_l}\right]
    =\sum_{b_1,\ldots,b_L\ge1}\left[\prod_{l=1}^L\mu_{g_l}^{b_l}\,\frac{b_l!\,s(c_l,b_l)}{c_l!}\right]z^{\,M-n}\,,
\end{align}
where $n:=\sum_l b_l$ (the sum over each $b_l$ is effectively capped at $c_l$, since $s(c_l,b_l)=0$ for $b_l>c_l$, and $\sum_l(c_l-b_l)=M-n$). Applying the linear functional $\langle\cdot\rangle_M$, i.e.\ $z^D\mapsto M^{\underline D}/M^D$, term by term, and using $M^{\underline{M-n}}=M(M-1)\cdots(n+1)=M!/n!$,
\begin{align}
    \frac{M^M}{M!}\Big\langle\prod_{l=1}^L a_{c_l}(\mu_{g_l},z)\Big\rangle_M
    &=\frac{M^M}{M!}\sum_{b_1,\ldots,b_L\ge1}\left[\prod_{l=1}^L\mu_{g_l}^{b_l}\,\frac{b_l!\,s(c_l,b_l)}{c_l!}\right]
    \frac{M^{\underline{M-n}}}{M^{M-n}}\nonumber\\
    &=\sum_{b_1,\ldots,b_L\ge1}\left[\prod_{l=1}^L\mu_{g_l}^{b_l}\,\frac{b_l!\,s(c_l,b_l)}{c_l!}\right]
    \frac{M^M}{n!\,M^{M-n}}\nonumber\\
    &=\sum_{b_1,\ldots,b_L\ge1}\frac{M^n}{n!}\prod_{l=1}^L\mu_{g_l}^{b_l}\,\frac{b_l!\,s(c_l,b_l)}{c_l!}
    \;=\;q(\alpha)\,,
\end{align}
the last equality by \eqref{eq:explicit}. Since this holds for every $\alpha$, \eqref{eq:pairing} coincides with \eqref{eq:explicit}, completing the proof of Theorem~\ref{thm:closed_form} in both stated forms. \hfill$\square$

\subsubsection{Proof of uniqueness}
Existence alone leaves open whether some other quasi-probability could satisfy the same order-$M$ constraint. We now show it cannot: the solution constructed above is the only one possible.

\begin{theorem}[Uniqueness]
\label{prop:uniqueness}
For every $M\ge1$ and every $\Gamma$, the quasi-probability $q:\Omega^M\to\mathbb{R}$ solving \eqref{eq:kth_order_constraint} is unique.
\end{theorem}

\begin{proof}
Since $\Phi(q):=\sum_{\alpha\in\Omega^M}q(\alpha)E_\alpha \bmod \mathds{A}_{\ge M+1}$ decomposes into homogeneous pieces $\Phi(q)=\sum_{k=0}^M\Phi_k(q)$ with $\Phi_k(q)\in\mathds{A}_k$, it suffices to show that the top piece $\Psi:=\Phi_M$ alone is injective: if $\Psi(q)=0$ forces $q=0$, then certainly $\Phi(q)=0$ does too. Concretely, $\Psi$ is the linear map from the vector space spanned by $\Omega^M$ to $\mathds{A}_M$ sending the basis vector $\alpha$ to $\Psi_\alpha:=[\deg=M]E_\alpha$, and $\Psi(q)=\sum_\alpha q(\alpha)\Psi_\alpha$.

\emph{Step 1: which words contaminate which.} Expanding each factor of $E_\alpha=e^{\mathcal{L}_{\alpha_M}}\cdots e^{\mathcal{L}_{\alpha_1}}$ in its own Taylor series and collecting the terms of total degree $M$,
\begin{align}
    \Psi_\alpha=\sum_{\substack{k_1,\ldots,k_M\ge0\\ \sum_i k_i=M}}\frac{1}{\prod_i k_i!}\,\mathcal{L}_{\alpha_M}^{k_M}\cdots\mathcal{L}_{\alpha_1}^{k_1}\,.
    \label{eq:uniq_expand}
\end{align}
Each term of \eqref{eq:uniq_expand} is itself a word: position $i$ is either \emph{dropped} ($k_i=0$, contributing nothing) or \emph{survives} with multiplicity $k_i\ge1$, contributing a block of $k_i$ copies of $\alpha_i$; the resulting word $\delta$ is the concatenation of the surviving blocks in their original order.

Let $\alpha$ have run decomposition $(g_1,c'_1),\ldots,(g_L,c'_L)$. If, in a given term of \eqref{eq:uniq_expand}, every one of these $L$ \textit{runs} retains at least one surviving position, no run is deleted and no two runs of different value are ever brought into contact, so the resulting word $\delta$ has run-\emph{values} exactly $(g_1,\ldots,g_L)$, in the same order — only the run-\emph{lengths} $(c_1,\ldots,c_L)$ (still summing to $M$, still all $\ge1$) can change. If instead some run is dropped entirely, the number of distinct values surviving strictly decreases, so $L(\delta)<L(\alpha)$. In either case:
\begin{align}
    L(\delta)\le L(\alpha)\,,\quad\text{with equality iff $\delta$ has the same run-value sequence $(g_1,\ldots,g_L)$ as $\alpha$.}
    \label{eq:uniq_triangular}
\end{align}
Since $\Psi$ is linear, it is represented, in the basis $\Omega^M$, by the matrix whose $(\delta,\alpha)$ entry is the coefficient of $\delta$ in $\Psi_\alpha$; by \eqref{eq:uniq_triangular}, this matrix is block-triangular when rows and columns are ordered by decreasing $L$, and, within the block of a fixed $L$, block-diagonal across run-value sequences: two words with the same $L$ can mix only if they share the same run-value sequence $(g_1,\ldots,g_L)$. It therefore suffices to show that each of these diagonal blocks — one for every $L$ and every run-value sequence of length $L$ — is itself invertible; block-triangularity with invertible diagonal blocks then makes $\Psi$ invertible. Fix one such block. Every word in it shares the same run-value sequence, so it is uniquely labeled by its length composition, and the block is naturally indexed by pairs of compositions $c,c'$ of $M$ into $L$ positive parts; we write $T_{c,c'}$ for its $(c,c')$ entry, i.e.\ the coefficient of the word with composition $c$ in $\Psi_\alpha$, where $\alpha$ denotes the word with composition $c'$ in this run-value sequence. Steps 2 and 3 compute $T$ explicitly and show it is always invertible.

\emph{Step 2: the per-run transform.} Fix $L$ and a run-value sequence $(g_1,\ldots,g_L)$, and let $\alpha,\delta$ be two words with exactly this run-value sequence, with length compositions $c'=(c'_1,\ldots,c'_L)$ and $c=(c_1,\ldots,c_L)$ respectively (both summing to $M$), so that $T_{c,c'}$ is precisely the coefficient of $\delta$ in $\Psi_\alpha$. By Step 1, only terms of \eqref{eq:uniq_expand} that keep every run alive can contribute, and for these the choice of how many surviving positions and what multiplicities each run uses is entirely independent run by run (dropped runs are excluded, and different runs never interact since no run is elsewhere deleted to bring them into contact). Hence the coefficient factorizes,
\begin{align}
    T_{c,c'}=\prod_{l=1}^L R(c'_l\to c_l)\,,
\end{align}
where $R(c'\to c)$ is the total weight of distributing a run of $c'$ original positions (all one label) into a surviving block of length $c$: choosing $m\ge1$ of the $c'$ positions to survive and splitting $c$ among them,
\begin{align}
    R(c'\to c)=\sum_{m=1}^{c'}\binom{c'}{m}\sum_{\substack{k_1,\ldots,k_m\ge1\\ \sum_j k_j=c}}\prod_j\frac1{k_j!}
    =\sum_{m=1}^{c'}\binom{c'}{m}\,[x^c](e^x-1)^m
    =[x^c]\,(1+(e^x-1))^{c'}=[x^c]e^{c'x}\,,
\end{align}
using the exponential generating function $\sum_{k\ge1}x^k/k!=e^x-1$ and the binomial theorem. Thus, for every $c\ge1$,
\begin{align}
    R(c'\to c)=\frac{(c')^{c}}{c!}\,,\qquad\text{so}\qquad T_{c,c'}=\prod_{l=1}^L\frac{(c'_l)^{c_l}}{c_l!}\,.
    \label{eq:uniq_Tformula}
\end{align}

\emph{Step 3: the block $T$ is invertible.} Fix $L$ and consider $T_{c,c'}$ as a matrix indexed by compositions $c,c'$ of $M$ into $L$ positive parts. Since $c'_l \ge 1$ for all $l$, we factor $T_{c,c'} = \frac{\prod_l c'_l}{\prod_l c_l!} \prod_{l=1}^L (c'_l)^{c_l-1}$, where $\prod_l c'_l > 0$ acts as a non-singular column scaling. It thus suffices to show that the matrix of monomials $M_{c,c'} := \prod_{l=1}^L (c'_l)^{c_l-1}$ has full rank. We exhibit, for every target composition $c^*$, an explicit linear combination of these degree $M-L$ monomials that isolates $c^*$. Define, for $x=(x_1,\ldots,x_L)$,
\begin{align}
    \Phi_{c^*}(x):=\prod_{l=1}^L\ \prod_{j=1}^{c^*_l-1}(x_l-j)\,,
\end{align}
a polynomial of total degree $\sum_l (c^*_l-1) = M-L$. Evaluating at $x=c'$ for any composition $c'\ne c^*$: since $\sum_l c'_l=\sum_l c^*_l=M$, there must exist some index $l$ with $c'_l < c^*_l$, causing the factor $(x_l-c'_l)$ to vanish; hence $\Phi_{c^*}(c')=0$. At $c'=c^*$, $\Phi_{c^*}(c^*)=\prod_{l=1}^L(c^*_l-1)!\ne 0$. On the constraint surface $\sum_l x_l = M$, $\Phi_{c^*}$ expands directly into a linear combination of the degree $M-L$ monomials $\prod_{l=1}^L x_l^{c_l-1}$. These combinations isolate each column $c^*$, proving that $T$ has full rank and is invertible.
Since this holds for every run-value sequence within every $L$, every diagonal block of $\Psi$'s matrix is invertible; by the block-triangularity established in Step 1, $\Psi$ itself is therefore invertible. In particular $\Psi(q)=0$ forces $q=0$, which is what was to be shown.
\end{proof}

Theorem~\ref{thm:closed_form} and Theorem~\ref{prop:uniqueness} together show that the order-$M$ sampling condition \eqref{eq:kth_order_constraint} determines a well-defined quasi-probability distribution $q$ for every choice of $M$ and $\Gamma$: it exists, by explicit construction via the generating function $\mathcal{G}(y)$, and it is the only distribution with this property, since the constraint is linear and its associated map is injective on degree-$M$ words. This is what justifies referring to $q(\alpha)$ in \eqref{eq:explicit} as the order-$M$ quasi-probability distribution throughout the main text.

\clearpage

\section{B --- Model-independent bound on the $\ell_1$-norm of the quasi-probability distribution}

\subsubsection{Setup}

Recall from Section~A that the order-$M$ quasi-probability $q:\Omega^M\to\mathbb{R}$ is the unique solution of \eqref{eq:kth_order_constraint}, with closed form \eqref{eq:explicit}. Its sampling overhead is controlled by
\begin{align}
    \|q_M\|_1:=\sum_{\alpha\in\Omega^M}|q(\alpha)|\,.
\end{align}
Setting $\mathcal{L}_\gamma\to0$ in \eqref{eq:kth_order_constraint} sends every $E_\alpha\to1$ and $e^{M\mathcal{L}}\to1$, so $\sum_{\alpha}q(\alpha)=1$ identically; by the triangle inequality this gives the trivial floor
\begin{align}
    \|q_M\|_1\ \ge\ 1\,,
    \label{eq:trivial_floor}
\end{align}
with equality iff $q(\alpha)\ge0$ for every $\alpha$.

\subsubsection{Model-independent upper bound}

\begin{theorem}[Universal bound]
\label{thm:l1bound}
For every $M\ge1$, every $\Gamma$, and every weights $\{\mu_\gamma\}_{\gamma\in\Omega}$,
\begin{align}
    \|q_M\|_1\ \le\ \binom{2M-1}{M}\,.
    \label{eq:l1bound}
\end{align}
\end{theorem}

\begin{proof}
Every factor in
\begin{align}
    q(\alpha)=\sum_{b_1,\ldots,b_L\ge1}\frac{M^n}{n!}\prod_{l=1}^L\mu_{g_l}^{b_l}\,\frac{b_l!\,s(c_l,b_l)}{c_l!}\,,\qquad n=\sum_l b_l\,,
\end{align}
is positive except the signed Stirling number $s(c_l,b_l)$, so the triangle inequality applied termwise gives $|q(\alpha)|\le\hat q(\alpha)$, where $\hat q(\alpha)$ is obtained by replacing $s(c_l,b_l)$ with the unsigned Stirling number $c(c_l,b_l):=|s(c_l,b_l)|$. Since $\ln^b(1+x)/b!=\sum_c s(c,b)x^c/c!$ and $(-\ln(1-x))^b/b!=\sum_c c(c,b)x^c/c!$ are governed by the same combinatorics with $\ln(1+x)$ replaced by $-\ln(1-x)$, the word-expansion argument of Section~A's existence proof carries over verbatim under this replacement, giving the exact identity
\begin{align}
    \sum_{\alpha\in\Omega^N}\hat q(\alpha)\,E_\alpha=[y^N]\exp\Big(-M\sum_{\gamma\in\Omega}\mu_\gamma\ln(1-ye^{\mathcal{L}_\gamma})\Big)\,,
\end{align}
with every $\hat q(\alpha)\ge0$. Evaluating at $\mathcal{L}_\gamma=0$ for every $\gamma$ (so $E_\alpha\to1$) collapses the $\gamma$-dependence entirely, since $\sum_\gamma\mu_\gamma=1$:
\begin{align}
    \sum_{\alpha\in\Omega^M}\hat q(\alpha)=[y^M]\exp\big(-M\ln(1-y)\big)=[y^M](1-y)^{-M}=\binom{2M-1}{M}\,.
\end{align}
Since $\|q_M\|_1=\sum_\alpha|q(\alpha)|\le\sum_\alpha\hat q(\alpha)$, this proves \eqref{eq:l1bound}.
\end{proof}

\subsubsection{The range of $\eta$}

\begin{corollary}
\label{cor:ceiling}
If $\|q_M\|_1$ grows as $\|q_M\|_1\sim\eta^{(M-1)/2}$ for some problem-specific rate $\eta$, then $\eta\in[1,16]$.
\end{corollary}

\begin{proof}
By Stirling's approximation, $\binom{2M-1}{M}=\tfrac12\binom{2M}{M}\sim\dfrac{4^M}{2\sqrt{\pi M}}$. Writing $4^M=16^{M/2}=4\cdot16^{(M-1)/2}$, this gives $\binom{2M-1}{M}=\Theta(M^{-1/2})\cdot16^{(M-1)/2}$, so $\binom{2M-1}{M}^{2/(M-1)}\to16$ as $M\to\infty$. Theorem~\ref{thm:l1bound} then forces $\eta\le16$. Combined with the trivial floor \eqref{eq:trivial_floor}, which gives $\eta\ge1$, this proves $\eta\in[1,16]$.
\end{proof}

The floor $\eta=1$ is saturated only in the degenerate case $q(\alpha)\ge0$ for all $\alpha$ (e.g.\ $\Gamma=1$); away from this trivial regime, $\|q_M\|_1$ can grow due to the non-commutativity of the Hamiltonian terms. The precise value of $\eta$ within $[1,16]$ depends on the full weight distribution $\{\mu_\gamma\}$ and is extracted numerically in the numerical results reported in the main text.

\clearpage\section{C --- Optimal choice of $M$ and $N$}

\subsubsection{Setup}

The sampling-based realization of an order-$M$ scheme analyzed in the previous Appendices 
incurs a shot-count overhead $\|q_M\|_1^2$ relative to an idealized (unsigned)
sampler, with $\|q_M\|_1\sim\eta^{(M-1)/2}$ observed numerically for a
problem-specific rate $\eta$ confined to $[1,16]$ by Corollary~\ref{cor:ceiling}. If the
total evolution is split into $N$ independent segments, each implementing an order-$M$
formula over a fraction $\Lambda T/N$ of the total weighted evolution time, the
overhead compounds multiplicatively across segments, and the number of circuit
repetitions needed to reach a fixed statistical error scales as
$\|q_M\|_1^{2N}\sim\eta^{N(M-1)}$. Weighted by the $NM$ elementary channels used per
shot, this gives the total resource cost
\begin{equation}
C(N,M) \;=\; N\,M\,\eta^{\,N(M-1)}\,,
\qquad \eta \in [1, \, 16]\,,
\label{eq:cost}
\end{equation}
subject to the accuracy constraint relating $N$, $M$, $\Lambda T$ and the target error
$\epsilon$,
\begin{equation}
\epsilon \;\ge\; \frac{(2\Lambda T)^{M+1}}{(M+1)!\,N^M}\,e^{2\Lambda T/N}\,.
\label{eq:accbound}
\end{equation}
Larger $N$ and larger $M$ both shrink the right-hand side of \eqref{eq:accbound} and
both increase $C$, so the optimum always saturates the constraint -- spending more
accuracy than required is wasteful -- and we work throughout with the equality
\begin{equation}
\epsilon \;=\; \frac{(2\Lambda T)^{M+1}}{(M+1)!\,N^M}\,e^{2\Lambda T/N}\,.
\label{eq:accEq}
\end{equation}
The goal of this section is to describe how $(N^*,M^*):=\arg\min C(N,M)$ subject to
\eqref{eq:accEq} is obtained numerically in general, and to extract the resulting
scaling of $C^*$ in the two limits $\epsilon\to0$ at fixed $\Lambda T$, and
$\Lambda T\to\infty$ at fixed $\epsilon$.

\subsubsection{Reducing the constraint to a single equation}

Equation~\eqref{eq:accEq} defines a curve in the $(N,M)$ plane which can be solved
explicitly for $N$ as a function of $M$. Writing it as $N^Me^{-2\Lambda T/N}=A(M)$
with $A(M):=(2\Lambda T)^{M+1}/[(M+1)!\,\epsilon]$, the substitution $z=1/N$,
$w=(2\Lambda T/M)z$ brings it to the canonical Lambert-$W$ form $we^w=\omega(M)$, with
\begin{equation}
\omega(M) \;=\; \frac{1}{M}\left[\frac{(M+1)!\,\epsilon}{2\Lambda T}\right]^{1/M}\,.
\label{eq:omega}
\end{equation}
Undoing the substitutions gives the closed form
\begin{equation}
N(M) \;=\; \frac{2\Lambda T}{M\,W_0(\omega(M))}\,,
\label{eq:NofM}
\end{equation}
so that
\begin{align}
    \mathrm{C}_\mathrm{run}(N(M),M)=\frac{2\Lambda T}{W_0(\omega(M))}
\end{align}
with $W_0$ the principal branch of the Lambert $W$ function, computable numerically
alongside $\ln\Gamma$ for any real $M>0$.

\begin{lemma}[Monotonicity and the range of $M$]
\label{lem:monotone}
$N(M)$ is strictly decreasing in $M$: a higher-order formula strictly reduces the
number of segments needed at fixed accuracy. Consequently there is a unique
$M_{\max}=M_{\max}(\epsilon,\Lambda T)$ with $N(M_{\max})=1$, defined by
\begin{equation}
\epsilon \;=\; \frac{(2\Lambda T)^{M_{\max}+1}\,e^{2\Lambda T}}{(M_{\max}+1)!}\,,
\label{eq:Mmaxeq}
\end{equation}
and the feasible portion of \eqref{eq:accEq} with $N\ge1$, $M\ge1$ is exactly the
compact interval $M\in[1,M_{\max}]$, on which $N(M)$ decreases monotonically from
$N(1)$ down to $1$.
\end{lemma}

\subsubsection{General numerical procedure}

Minimizing $\mathrm{C}$ is equivalent to minimizing $\ln \mathrm{C}=\ln N+\ln M+N(M-1)\ln\eta$ over the compact interval $M \in [1, \,M_{\max}]$. The minimum is attained either at an interior stationary point or at one of the two
endpoints $M=1$ or $M=M_{\max}$. Enforcing stationarity of $\ln C$ along the
constraint $\Phi(N,M):=(M+1)\ln(2\Lambda T)-\ln\Gamma(M+2)-M\ln N+2\Lambda T/N-\ln\epsilon=0$
via a Lagrange multiplier and eliminating the multiplier between the two resulting
equations gives one exact relation between $N$ and $M$,
\begin{equation}
\left[\frac1N+(M-1)\ln\eta\right]\Big[\ln(2\Lambda T)-\psi(M+2)-\ln N\Big]
=
\left[\frac1M+N\ln\eta\right]\left[-\frac{MN+2\Lambda T}{N^2}\right]\,,
\label{eq:eqA}
\end{equation}
with $\psi$ the digamma function; no approximation has been used. Equation
\eqref{eq:eqA}, together with $\Phi(N,M)=0$, is a $2\times2$ nonlinear system solvable
by standard methods (e.g.\ Newton's method) for any interior candidate $(N_c,M_c)$.

\begin{proposition}[Exact optimum]
\label{prop:exact}
The optimum of $\mathrm{C}(N,M)$ subject to \eqref{eq:accEq}, relaxed to real
$M\in[1,M_{\max}]$ with $N=N(M)$ from \eqref{eq:NofM}, is
\begin{equation}
M^*=\operatorname*{arg\,min}_{M\,\in\,\{1,\,M_c,\,M_{\max}\}}\mathrm{C}\big(N(M),M\big)\,,
\qquad N^*=N(M^*)\,,
\label{eq:exactproc}
\end{equation}
where $M_c$ is any root of \eqref{eq:eqA} lying in $(1,M_{\max})$, discarded if none
exists, $M_{\max}$ solves \eqref{eq:Mmaxeq}, and $\mathrm{C}(N(1),1)=N(1)=2\Lambda T/W_0(\epsilon/\Lambda T)$.
\end{proposition}

This is a complete numerical prescription for any $\epsilon,\Lambda T,\eta$: solve the
single transcendental equation \eqref{eq:Mmaxeq} for $M_{\max}$, solve the $2\times2$
system \eqref{eq:eqA}$+\Phi=0$ for any interior root, evaluate $C$ at the (at most
three) candidates, and keep the smallest; rounding $M^*$ to the nearest feasible
integer and re-optimizing $N$ there refines this to the literal integer solution
without affecting the asymptotics below.

\subsubsection{Limit $\epsilon\to0$ at fixed $\Lambda T$}

For fixed $\Lambda T$, the interior stationarity system \eqref{eq:eqA} is a small
perturbation of the simpler problem obtained by replacing $N(M-1)\to NM$ in the cost
exponent, whose minimizer is known to satisfy $M\to\infty$ while $\mathrm{C}_\mathrm{run}=NM$ approaches
the finite limit $2\Lambda T/W_0(1/e)$ as $\epsilon\to0$ -- and hence
$N=\mathrm{C}_\mathrm{run}/M\to0$. Since $N<1$ is infeasible, the interior candidate $M_c$ has simply
moved past $M_{\max}$: the constrained system has no admissible root in
$[1,M_{\max}]$, and by Proposition~\ref{prop:exact} the minimum sits at the right
endpoint,
\begin{equation}
N^*=1\,,\qquad M^*=M_{\max}(\epsilon,\Lambda T)\,.
\label{eq:eps0result}
\end{equation}
Solving \eqref{eq:Mmaxeq} asymptotically (Stirling's approximation on $\ln\Gamma$,
followed by the same Lambert-$W$ substitution as above) gives
\begin{equation}
M_{\max}+1 \sim \frac{\ln(1/\epsilon)}{W_0\!\left(\dfrac{\ln(1/\epsilon)}{2\Lambda T\,e}\right)}\,.
\label{eq:Mmaxasym}
\end{equation}
Since $N^*=1$, then $\mathrm{C}^*=M^*\eta^{M^*-1}$, so
\begin{equation}
\ln \mathrm{C}^*(\epsilon,\Lambda T) \;\sim\; \ln\eta\cdot\frac{\ln(1/\epsilon)}{W_0\!\left(\dfrac{\ln(1/\epsilon)}{2\Lambda T\,e}\right)}\,,
\label{eq:eps0scaling}
\end{equation}
i.e.\ $\mathrm{C}^*$ grows quasi-polynomially in $1/\epsilon$: slower than any fixed power of
$1/\epsilon$, but faster than any fixed power of $\ln(1/\epsilon)$.

\subsubsection{Limit $\Lambda T\to\infty$ at fixed $\epsilon$}

Here the mechanism differs. For fixed $M\ge2$, \eqref{eq:accEq} forces
$N(M)\sim(\Lambda T)^{1+1/M}$ as $\Lambda T\to\infty$, so it is natural to ask first
whether the interior stationary point of Proposition~\ref{prop:exact} -- the pair
$(M_c,N_c)$ obtained by solving the two Lagrange stationarity conditions for
$\mathrm{C}(N,M)=NM\,\eta^{N(M-1)}$ on the accuracy-constraint surface, treating $M$
and $N$ as continuous variables -- remains feasible in this regime. It does: solving
those conditions asymptotically gives $M_c\sim a\,(\Lambda T)^{2/3}$ and
$N_c\sim b\,(\Lambda T)^{1/3}$, for constants $a,b$ depending only on $\eta$, and both
lie comfortably inside the admissible range $(1,M_{\max})$, since $M_{\max}\sim
d\,\Lambda T$ grows linearly and so outpaces $M_c$ for large $\Lambda T$.
Feasibility alone does not decide the optimum, however, because the three candidates
that Proposition~\ref{prop:exact} instructs us to compare -- the interior point
$(M_c,N_c)$ and the two endpoints $M=1$ and $M=M_{\max}$ -- scale very differently
once their costs are worked out. At the interior point, both the product $N_cM_c$ and
the exponent $N_c(M_c-1)$ entering $\mathrm{C}$ grow, to leading order, as
$(2/W_0(1/e))\,\Lambda T$: linearly in $\Lambda T$, so that
$\mathrm{C}(N_c,M_c)\sim\eta^{(2/W_0(1/e))\Lambda T}$ is exponential in $\Lambda T$
whenever $\eta>1$. The endpoint $M=M_{\max}$ fares no better, and for the same
reason: since $N(M_{\max})=1$ by definition of $M_{\max}$, the cost there is
$\mathrm{C}(1,M_{\max})=M_{\max}\,\eta^{M_{\max}-1}\sim\eta^{d\,\Lambda T}$, again
exponential in $\Lambda T$. Only at the remaining endpoint, $M=1$, does the
exponential factor disappear entirely: the exponent $N(M-1)$ vanishes identically
at $M=1$, so $\mathrm{C}(N(1),1)=N(1)$ outright, and using $W_0(x)\sim x$ for small
$x$ in \eqref{eq:accEq},
\begin{equation}
N(1)=\frac{2\Lambda T}{W_0(\epsilon/\Lambda T)}\sim\frac{2(\Lambda T)^2}{\epsilon}\,,
\label{eq:N1asym}
\end{equation}
which is purely polynomial in $\Lambda T$, with no exponential factor at all.
Since a polynomial remains smaller than an exponential for $\Lambda T$ large enough, then for any fixed $\eta>1$ there is a crossover scale
$\Lambda T_c(\epsilon,\eta)$ beyond which $\mathrm{C}(N(1),1)$ is the smallest of the
three candidates, so that for all $\Lambda T\gg\Lambda T_c$,
\begin{equation}
M^*=1\,,\qquad N^*=N(1)\sim\frac{2(\Lambda T)^2}{\epsilon}\,,\qquad
\mathrm{C}^*(\Lambda T)\sim\frac{2(\Lambda T)^2}{\epsilon}\,,
\label{eq:Tinfresult}
\end{equation}
so that the optimal cost itself is polynomial -- quadratic -- in $\Lambda T$.

The degenerate case $\eta=1$ is qualitatively different, and it is the one case in
which the interior point genuinely wins. When $\eta=1$, the exponential factor in
$\mathrm{C}$ is identically $1$, so $\mathrm{C}(N,M)=NM$ exactly and there is no
exponential term at any $M$ to compete with; the interior stationary point is then
the true optimum, with $M^*\sim a\,(\Lambda T)^{2/3}$, $N^*\sim b\,(\Lambda T)^{1/3}$,
and $\mathrm{C}^*=N^*M^*\sim(2/W_0(1/e))\,\Lambda T$ -- linear in $\Lambda T$, and
therefore cheaper than the polynomial-but-quadratic result of \eqref{eq:Tinfresult}.
The two limits $\eta\to1^+$ and $\Lambda T\to\infty$ thus do not commute: for any
fixed $\eta>1$, however close to $1$, taking $\Lambda T\to\infty$ first always drives
the optimum to $M^*=1$, never to the interior point that is optimal exactly at
$\eta=1$.

\end{document}